\documentclass[12pt]{amsart}
\usepackage{amssymb}
\usepackage{color}
\usepackage{amsmath,epic,curves,amscd}
\usepackage[english]{babel}
\usepackage{graphicx}
\usepackage{comment}
\usepackage{appendix}
\usepackage{mathdots}
\usepackage[all]{xy}
\newtheorem{claim}{}[section]
\newtheorem{theorem}[claim]{Theorem}

\theoremstyle{remark}

\renewenvironment{proof}{\noindent{\it Proof. \hskip0pt}}
                      {$\square$\par\medskip}

\begin{document}
\baselineskip 6.0 truemm
\parindent 1.5 true pc

\newcommand\xx{{\text{\sf X}}}
\newcommand\lan{\langle}
\newcommand\ran{\rangle}
\newcommand\tr{{\text{\rm Tr}}\,}
\newcommand\ot{\otimes}
\newcommand\ol{\overline}
\newcommand\join{\vee}
\newcommand\meet{\wedge}
\renewcommand\ker{{\text{\rm Ker}}\,}
\newcommand\image{{\text{\rm Im}}\,}
\newcommand\id{{\text{\rm id}}}
\newcommand\tp{{\text{\rm tp}}}
\newcommand\pr{\prime}
\newcommand\e{\epsilon}
\newcommand\la{\lambda}
\newcommand\inte{{\text{\rm int}}\,}
\newcommand\ttt{{\text{\rm t}}}
\newcommand\spa{{\text{\rm span}}\,}
\newcommand\conv{{\text{\rm conv}}\,}
\newcommand\rank{\ {\text{\rm rank of}}\ }
\newcommand\re{{\text{\rm Re}}\,}
\newcommand\ppt{\mathbb T}
\newcommand\rk{{\text{\rm rank}}\,}
\newcommand\SN{{\text{\rm SN}}\,}
\newcommand\SR{{\text{\rm SR}}\,}
\newcommand\HA{{\mathcal H}_A}
\newcommand\HB{{\mathcal H}_B}
\newcommand\HC{{\mathcal H}_C}
\newcommand\CI{{\mathcal I}}
\newcommand{\bra}[1]{\langle{#1}|}
\newcommand{\ket}[1]{|{#1}\rangle}
\newcommand\cl{\mathcal}
\newcommand\idd{{\text{\rm id}}}
\newcommand\OMAX{{\text{\rm OMAX}}}
\newcommand\OMIN{{\text{\rm OMIN}}}
\newcommand\diag{{\text{\rm Diag}}\,}
\newcommand\calI{{\mathcal I}}
\newcommand\bfi{{\bf i}}
\newcommand\bfj{{\bf j}}
\newcommand\bfk{{\bf k}}
\newcommand\bfl{{\bf l}}
\newcommand\bfp{{\bf p}}
\newcommand\bfq{{\bf q}}
\newcommand\bfzero{{\bf 0}}
\newcommand\bfone{{\bf 1}}
\newcommand\sa{{\rm sa}}
\newcommand\ph{{\rm ph}}
\newcommand\phase{{\rm ph}}
\newcommand\res{{\text{\rm res}}}
\newcommand{\algname}[1]{{\sc #1}}
\newcommand{\Setminus}{\setminus\hskip-0.2truecm\setminus\,}
\newcommand\calv{{\mathcal V}}
\newcommand\calg{{\mathcal G}}
\newcommand\calt{{\mathcal T}}
\newcommand\calvnR{{\mathcal V}_n^{\mathbb R}}
\newcommand\D{{\mathcal D}}
\newcommand\C{{\mathcal C}}
\newcommand\aaa{\alpha}
\newcommand\bbb{\beta}
\newcommand\ccc{\gamma}
\newcommand\tefrac{\textstyle\frac}
\newcommand\xxx{\xx^\sigma}
\newcommand\E{{\mathcal E}}
\newcommand\W{{\mathcal W}}
\newcommand\X{{\mathcal X}}
\newcommand\variable{{\mathcal X}}

\title{Partial separability/entanglement violates distributive rules}

\author{Kyung Hoon Han, Seung-Hyeok Kye and Szil\'ard Szalay}
\address{Kyung Hoon Han, Department of Data Science, The University of Suwon, Gyeonggi-do 445-743, Korea}
\email{kyunghoon.han at gmail.com}
\address{Seung-Hyeok Kye, Department of Mathematics and Institute of Mathematics, Seoul National University, Seoul 151-742, Korea}
\email{kye at snu.ac.kr}
\address{Szil\'ard Szalay, Wigner Research Centre for Physics, 29-33, Konkoly-Thege Mikl\' os, H-1121 Budapest, Hungary}
\email{szalay.szilard at wigner.mta.hu}



\begin{abstract}
We found three qubit Greenberger-Horne-Zeilinger diagonal states which tell us that
the partial separability of three qubit states violates the distributive rules with respect to
the two operations of convex sum and intersection. The gaps between the convex sets involving the
distributive rules are of nonzero volume.
\end{abstract}

\maketitle

\section{Introduction}

Pure states in classical probability theory are uncorrelated, which
is not the case in quantum probability theory, where this
nonclassical form of correlation is called entanglement
\cite{Horodecki-2009}. Beyond the conceptual questions it raises
\cite{Schrodinger-1935a,Schrodinger-1935b}, entanglement plays a key
role in the physics of strongly correlated many-body systems
\cite{Amico-2008}, and also finds direct applications in quantum
information theory \cite{Nielsen-2000}. Mixed states in classical
probability theory arise as statistical mixtures (convex
combinations) of pure, hence uncorrelated states. Again, this is not
the case in quantum probability theory, and states which are
mixtures of uncorrelated states are called separable, while the
others are entangled \cite{Werner-1989}.

In the case of multipartite systems, the partitions of the total
system into subsystems give rise to various notions of partial
separability. In the tripartite case with the elementary subsystems
$A$, $B$ and $C$, we have three nontrivial partitions $A$-$BC$,
$B$-$CA$ and $C$-$AB$, and the corresponding partial separability
properties are called $A$-$BC$-separability, $B$-$CA$-separability
and $C$-$AB$-separability, respectively. We call these {\sl basic}
biseparabilities.

It is natural to consider the intersections (also called partial
separability classes) and convex hulls of the three convex sets
consisting of the above three kinds of basic biseparable states. For
example, the intersection of them \cite{bdmsst}, the intersections
of two of them \cite{{dur-cirac},{dur-cirac-tarrach}} and the convex
hull of them \cite{acin, seevinck-uffink} have been considered. More
recently, the convex hulls of two of them have also been considered
together with intersections and complements of convex sets arising
in the way \cite{{sz2012},{han_kye_bisep_exam},{han_kye_pe}},
leading to the description of the hierarchy of the intersections
\cite{{sz2015}}. See also \cite{{sz2011},{sz2018},{Szalay-2019}} for
further developments. Recall that the intersection and convex hull
of the convex sets of three basic biseparable states give rise to
{\sl fully biseparable} and {\sl biseparable} states, respectively.
Tripartite states which are not biseparable are called
genuinely multipartite entangled.

In this context, we consider the lattice generated by three convex sets of three qubit basic biseparable states
with respect to the above mentioned two operations, intersection and convex hull. In this way, we
deal with convex hulls of intersections, as well as intersections of convex hulls of convex sets
arising from basic biseparability. One may go further
to investigate the whole structures of partial separability and partial entanglement.
Due to technical reasons,
we will consider the three convex cones $\aaa$, $\bbb$ and $\ccc$ of all {\sl un-normalized}
$A$-$BC$, $B$-$CA$ and $C$-$AB$ biseparable three qubit states, respectively.
We note that the convex hull $\sigma\join\tau$ of two convex cones $\sigma$ and $\tau$ coincides with the sum $\sigma+\tau$ of them.
The intersection of two convex cones $\sigma$ and $\tau$ will be denoted by $\sigma\meet\tau$ following the lattice notations.
In general, the two operations $\meet$ and $\join$ among convex sets obey
associative rule and commutative rule. Furthermore, they also satisfy
the following relations
$$
\begin{aligned}
\sigma\join \sigma=\sigma,\qquad
&\sigma\meet \sigma=\sigma,\\
(\sigma\join \tau)\meet \sigma=\sigma,\qquad &(\sigma\meet \tau)\join \sigma=\sigma,
\end{aligned}
$$
and so they give rise to a lattice.

A lattice $(\mathcal L,\join,\meet)$ is distributive if the identities
$$
x \meet (y \join z) = (x \meet y) \join (x \meet z), \qquad x \join (y \meet z) = (x \join y) \meet (x \join z)
$$
hold for all $x, y, z \in \mathcal L$.
In general, the inequalities
$$
x \meet (y \join z) \ge (x \meet y) \join (x \meet z), \qquad x \join (y \meet z) \le (x \join y) \meet (x \join z)
$$
always hold trivially.
We denote by ${\mathcal L}$ the lattice generated by three convex cones
$\aaa$, $\bbb$ and $\ccc$. Therefore, ${\mathcal L}$ is the smallest lattice containing
the convex cones $\aaa$, $\bbb$ and $\ccc$ in the $64$-dimensional real vector space of all self-adjoint three qubit
matrices. The purpose of this note is to show that the lattice ${\mathcal L}$ does not satisfy
the distributive rules. More precisely, we show that both inequalities
\begin{align}
(\aaa\meet\bbb)\join(\aaa\meet\ccc)&\le \aaa\meet(\bbb\join\ccc),\label{dist-ineq}\\
\bbb\join(\ccc\meet\aaa)&\le (\bbb\join\ccc)\meet(\bbb\join\aaa)\label{dist-ineq-2}
\end{align}
are strict.
Furthermore, the gaps between the two sets are of nonzero volume, in both cases.
We also show that the lattice ${\mathcal L}$ does not satisfy the modularity
which is weaker than distributivity.

We note that a state $\varrho$ in the gap
$\aaa\meet(\bbb\join\ccc)\setminus (\aaa\meet\bbb)\join(\aaa\meet\ccc)$
by the strict inequality in (\ref{dist-ineq}) has the following properties:
\begin{itemize}
\item
$\varrho$ is $A$-$BC$ biseparable,
and it is a mixture of $B$-$CA$ and $C$-$AB$ biseparable states,
\item
but, it is not a mixture of a simultaneously $A$-$BC$ and $B$-$CA$ biseparable state and
a simultaneously $A$-$BC$ and $C$-$AB$ biseparable state.
\end{itemize}
We will find such states among GHZ diagonal states.
On the other hand, a state $\varrho$ arising by the strict inequality in (\ref{dist-ineq-2}) has the following properties:
\begin{itemize}
\item
$\varrho$ is a mixture of $B$-$CA$ and $C$-$AB$ biseparable states, and it is also a mixture
of a $B$-$CA$ and $A$-$BC$ biseparable states,
\item
but it is not a mixture of $B$-$CA$ biseparable state and a simultaneously $C$-$AB$ and $A$-$BC$
biseparable state.
\end{itemize}

Examples will be found among GHZ diagonal states. After we provide backgrounds for this in the next section, we find analytic examples in Section 3.
We also consider in Section 4 the lattices arising from partial separability in general multi-partite systems to see that both distributivity and modularity are violated in them, too. We close the paper to ask several questions.

The authors are grateful to the referee for bringing our attention to the general multi-partite cases.

\section{$\xx$-shaped states}

We will find required examples among so called $\xx$-shaped states whose entries are zeros except for
diagonal and anti-diagonal entries by definition.
A self-adjoint  $\xx$-shaped three qubit matrix
is of the form
$$
\xx(a,b,z)= \left(
\begin{matrix}
a_1 &&&&&&& z_1\\
& a_2 &&&&& z_2 & \\
&& a_3 &&& z_3 &&\\
&&& a_4&z_4 &&&\\
&&& \bar z_4& b_4&&&\\
&& \bar z_3 &&& b_3 &&\\
& \bar z_2 &&&&& b_2 &\\
\bar z_1 &&&&&&& b_1
\end{matrix}
\right),
$$
with $a,b\in\mathbb R^4$ and $z\in\mathbb C^4$.
Here, $\mathbb C^2\ot\mathbb C^2\ot \mathbb C^2$ is identified with
the vector space $\mathbb C^8$ with respect to the lexicographic order of indices.
Note that $\xx(a,b,z)$ is a state if and only if $a_i,b_i\ge 0$ and $\sqrt{a_ib_i}\ge |z_i|$ for each $i=1,2,3,4$.

We recall that every  GHZ diagonal states \cite{GHZ} are in this form, and an \xx-state $\xx(a,b,z)$
is GHZ diagonal if and only if $a=b$ and $z\in\mathbb R^4$. In this case, we use the notation
$$
\xx\left(\begin{matrix}a\\z\end{matrix}\right)=\xx(a,a,z).
$$

By a {\sl pair} $\{i,j\}$, we will mean an unordered set with two distinct elements for simplicity throughout this paper.
For a given three qubit $\xx$-shaped state $\varrho=\xx(a,b,z)$, we
consider the inequalities
\begin{center}
\framebox{
\parbox[t][1.9cm]{14.80cm}{
\addvspace{0.1cm} \centering
$$
\begin{array}{ll}
S_1[i,j]:&\quad \min\{\sqrt{a_ib_i},\sqrt{a_jb_j}\}\ge \max\{|z_i|,|z_j|\},\\
S_2[i,j]:&\quad
\min\left\{\sqrt{a_ib_i}+\sqrt{a_jb_j},\sqrt{a_k b_k}+\sqrt{a_\ell b_\ell}\right\}
   \ge\max\left\{|z_i|+|z_j|,|z_k|+|z_\ell|\right\},\\
S_3 :&\quad \textstyle{\sum_{j\neq i}\sqrt{a_jb_j}\ge |z_i|},\quad i=1,2,3,4
\end{array}
$$
}}
\end{center}\medskip
for a pair $\{i,j\}$, where $\{k,\ell\}$ is chosen so that $\{i,j,k,l\}=\{1,2,3,4\}$. By
\cite[Proposition 5.2]{han_kye_optimal}, we have the following:
\begin{itemize}
\item
$\varrho\in\aaa$ if and only if $S_1[1,4]$ and $S_1[2,3]$ hold,
\item
$\varrho\in\bbb$ if and only if $S_1[1,3]$ and $S_1[2,4]$ hold,
\item
$\varrho\in\ccc$ if and only if $S_1[1,2]$ and $S_1[3,4]$ hold.
\end{itemize}
We also have the following
\begin{itemize}
\item
$\varrho\in\bbb\join\ccc$ if and only if $S_2[1,4]$ (equivalently $S_2[2,3]$) holds,
\item
$\varrho\in\ccc\join\aaa$ if and only if $S_2[1,3]$ (equivalently $S_2[2,4]$) holds,
\item
$\varrho\in\aaa\join\bbb$ if and only if $S_2[1,2]$ (equivalently $S_2[3,4]$) holds.
\end{itemize}
by \cite[Theorem 5.5]{han_kye_pe}.
The inequality $S_3$ will not be used in this paper, but
it is the characteristic inequality for the convex cone $\aaa\join\bbb\join\ccc$
\cite{{gao},{han_kye_optimal},{han_kye_pe},{Rafsanjani}}.

We will consider the above inequality $S_2$ for arbitrary
two pairs $\{i,j\}$ and $\{k,\ell\}$
as follows:
\begin{center}
\framebox{
\parbox[t][0.7cm]{14.80cm}{
\addvspace{0.1cm} \centering
$S_4[i,j|k,\ell]:\quad
\min\left\{\sqrt{a_ib_i}+\sqrt{a_jb_j},\sqrt{a_k b_k}+\sqrt{a_\ell b_\ell}\right\}
   \ge\max\left\{|z_i|+|z_j|,|z_k|+|z_\ell|\right\}$.
}}
\end{center}\medskip
If $\{i,j\}=\{k,\ell\}$ then the inequality $S_4[i,j|k,\ell]$
holds automatically for any {\sf X}-states $\varrho=\xx(a,b,z)$. If $\{i,j\}\cap\{k,\ell\}=\emptyset$
then the three inequalities $S_2[i,j]$,  $S_2[k,\ell]$ and $S_4[i,j|k,\ell]$ are identical.
In the other cases, the resulting inequalities are new ones.

In order to get required examples in the gaps between convex cones in the inequalities
(\ref{dist-ineq}) and (\ref{dist-ineq-2}),
we proceed to characterize the following convex cones
\begin{equation}\label{target}
\aaa\join(\bbb\meet\ccc),\qquad
\bbb\join(\ccc\meet\aaa),\qquad
\ccc\join(\aaa\meet\bbb).
\end{equation}
For this purpose, we will use the duality among
convex cones in a real vector space with a bi-linear pairing $\lan\ ,\ \ran$.
For a convex cone $C$, the dual cone $C^\circ$ is defined by
$$
C^\circ=\{x\in V:\lan x,y\ran\ge 0\ {\text{\rm for each}}\ y\in C\}.
$$
We are now working in the real vector spaces of all three qubit self-adjoint $\xx$-shaped matrices, where
the bi-linear pairing is defined by $\lan x,y\ran=\tr(yx^\ttt)$, as usual. See \cite{han_kye_pe}.
Every closed convex cone $C$ satisfies the relation $C=(C^\circ)^\circ$, which tells us that
$x\in C$ if and only if $\lan x,y\ran\ge 0$ for every $y\in C^\circ$.

The dual cones of the cones in (\ref{target}) have also been
characterized in \cite{han_kye_pe}. For a given {\sf X}-shaped
self-adjoint matrix $W=\xx(s,t,u)$ with $s_i,t_i\ge 0$ and
$u\in\mathbb C^4$, we have considered the inequalities $W_1, W_2,
W_3$ given by
\begin{center}
\framebox{
\parbox[t][1.9cm]{10.00cm}{
\addvspace{0.1cm} \centering
$$
\begin{array}{ll}
W_1[i,j]: &\quad \sqrt{s_it_i}+\sqrt{s_jt_j}\ge |u_i|+|u_j|,\\
W_2[i,j]: &\quad \sum_{k\neq j}\sqrt{s_kt_k}\ge |u_i|,\quad \sum_{k\neq i}\sqrt{s_kt_k}\ge |u_j|,\\
W_3 :&\quad \sum_{i=1}^4\sqrt{s_it_i}\ge\sum_{i=1}^4|u_i|,
\end{array}
$$
}}
\end{center}\medskip
for a pair $\{i,j\}$.
Then we have the following by \cite[Proposition 3.3]{han_kye_pe}:
\begin{itemize}
\item
$W\in\aaa^\circ$ if and only if $W_1[1,4]$ and $W_1[2,3]$ hold,
\item
$W\in\bbb^\circ$ if and only if $W_1[1,3]$ and $W_1[2,4]$ hold,
\item
$W\in\ccc^\circ$ if and only if $W_1[1,2]$ and $W_1[3,4]$ hold.
\end{itemize}
On the other hand, we also have the following by \cite[Theorem 5.2]{han_kye_pe}:
\begin{itemize}
\item
$W\in\bbb^\circ\join\ccc^\circ$ if and only if $W_2[1,4]$, $W_2[2,3]$ and $W_3$ hold,
\item
$W\in\ccc^\circ\join\aaa^\circ$ if and only if $W_2[1,3]$, $W_2[2,4]$ and $W_3$ hold,
\item
$W\in\aaa^\circ\join\bbb^\circ$ if and only if $W_2[1,2]$, $W_2[3,4]$ and $W_3$ hold.
\end{itemize}

For given $a,b\in\mathbb R^4$ with nonzero entries, and $z\in\mathbb C^4$ with $\arg z_i=\theta_i$,
we consider the following self-adjoint matrices:
$$
W_{[i,j|k,\ell]}=
\xx\left(\sqrt{b_i \over a_i} E_i + \sqrt{b_j \over a_j} E_j ,
\sqrt{a_i \over b_i} E_i + \sqrt{a_j \over b_j} E_j ,
-e^{- {\rm i} \theta_k} E_k - e^{- {\rm i} \theta_\ell} E_l\right)
$$
for pairs $\{i,j\}$ and $\{k,\ell\}$, where $\{E_i\}$ is the usual orthonormal basis of $\mathbb C^4$.
Then the inequality
$$
\lan W_{[i,j|k,\ell]}, \xx(a,b,z)\ran\ge 0
$$
gives rise to the inequality $\sqrt{a_i b_i} + \sqrt{a_j b_j} \ge |z_k| + |z_\ell|$. Now, we present the main result.
By \cite[Proposition 2.2]{han_kye_pe}, the conditions give rise to necessary criteria
for general states to belong to classes in (\ref{target})
in terms of diagonal and anti-diagonal entries.

\begin{theorem}\label{join-statestep3}
For a given three qubit {\sf X}-state $\varrho=\xx(a,b,z)$, we have the
following:
\begin{enumerate}
\item[(i)]
$\varrho\in\aaa\join(\bbb\meet\ccc)$ if and only if $S_4[i,j|k,\ell]$ holds
whenever $\{i,j\}$, $\{k,\ell\}$ are two of $\{1,2\}$, $\{1,3\}$, $\{2,4\}$, $\{3,4\}$;
\item[(ii)]
$\varrho\in\bbb\join(\ccc\meet\aaa)$ if and only if $S_4[i,j|k,\ell]$ holds
whenever $\{i,j\}$, $\{k,\ell\}$ are two of $\{1,2\}$, $\{1,4\}$, $\{2,3\}$, $\{3,4\}$;
\item[(iii)]
$\varrho\in\ccc\join(\aaa\meet\bbb)$ if and only if $S_4[i,j|k,\ell]$ holds
whenever $\{i,j\}$, $\{k,\ell\}$ are two of $\{1,3\}$, $\{1,4\}$, $\{2,3\}$, $\{2,4\}$.
\end{enumerate}
\end{theorem}

\begin{proof}
We will prove (i). For the \lq only if\rq\ part, we first consider the case $a_i,b_i> 0$ for $i=1,2,3,4$. Then
we see that $W_{[i,j|k,\ell]}$ belongs to $\aaa^\circ\meet(\bbb^\circ\join\ccc^\circ)$
whenever $\{i,j\}$, $\{k,\ell\}$ are two of $\{1,2\}$, $\{1,3\}$, $\{2,4\}$, $\{3,4\}$,
by checking the required inequalities $W_1[1,4]$, $W_1[2,3]$, $W_2[1,4]$, $W_2[2,3]$ and $W_3$.
Now, the inequality $\lan W_{[i,j|k,\ell]}, \varrho\ran\ge 0$
gives rise to $S_4[i,j|k,\ell]$ for such $\{i,j\}$, $\{k,\ell\}$. If $a$ or $b$ has a zero entry, then we consider
$\varrho+\varepsilon\,\xx(\bf 1,\bf 1,\bf 0)$ with arbitrary $\varepsilon >0$ and take $\varepsilon\to 0$ for the conclusion,
where ${\bf 1}=(1,1,1,1)$ and ${\bf 0}=(0,0,0,0)$.

For the converse,
it suffices to show the following:
\begin{itemize}
\item
if $\varrho=\xx(a,b,z)$ satisfies $S_4[i,j|k,\ell]$ whenever $\{i,j\}$, $\{k,l\}$ are two of $\{1,2\}$, $\{1,3\}$, $\{2,4\}$, $\{3,4\}$, and
$W=\xx(s,t,u)$ satisfies $W_1[1,4]$, $W_1[2,3]$, $W_2[1,4]$, $W_2[2,3]$ and $W_3$, then
$\lan W,\varrho\ran\ge 0$,
\end{itemize}
by  Proposition 2.2 and Corollary 2.3 of \cite{han_kye_pe}.
If $\varrho\in\aaa$ then there is nothing to prove, and so we may assume that $\varrho\notin\aaa$.
Without loss of generality, we may also assume that $|z_4|> \sqrt{a_1b_1}$. Putting $p=\max\{|z_2|, |z_3|\}$, we have
$$
\min \{ \sqrt{a_2b_2}, \sqrt{a_3b_3} \} \ge |z_4| + \left(p -\sqrt{a_1b_1}\right),
$$
by the inequalities $S_4[i,j|k,\ell]$ for $\{i,j\}=\{1,2\}, \{1,3\}$ and $\{k,\ell\}=\{2,4\},\{3,4\}$.
Therefore, we have
$$
\begin{aligned}
\sum_{i=2}^4\sqrt{s_it_i}\sqrt{a_ib_i}-|u_4||z_4|
&\ge(\sqrt{s_2t_2} + \sqrt{s_3t_3}) \min\{\sqrt{a_2b_2},\sqrt{a_3b_3}\} + \sqrt{s_4t_4} |z_4| - |u_4| |z_4|\\
&\ge\left(\sum_{i=2}^4\sqrt{s_it_i}-|u_4|\right)|z_4|
   +(\sqrt{s_2t_2}+\sqrt{s_3t_3})\left(p-\sqrt{a_1b_1}\right)\\
&\ge\left(\sum_{i=2}^4\sqrt{s_it_i}-|u_4|\right)\sqrt{a_1b_1}
   +(\sqrt{s_2t_2}+\sqrt{s_3t_3})\left(p-\sqrt{a_1b_1}\right)\\
&= (\sqrt{s_4t_4} - |u_4|)\sqrt{a_1b_1} +(\sqrt{s_2t_2} + \sqrt{s_3t_3})p,
\end{aligned}
$$
by the inequality $W_2[4,1]$ and the assumption $|z_4|> \sqrt{a_1b_1}$. We also have
$$
\begin{aligned}
\sqrt{s_1t_1}\sqrt{a_1b_1}-\sum_{i=1}^3|u_i||z_i|
&\ge \sqrt{s_1t_1}\sqrt{a_1b_1}-|u_1|\sqrt{a_1b_1}-|u_2||z_2|-|u_3||z_3|\\
&\ge \sqrt{s_1t_1}\sqrt{a_1b_1}-|u_1|\sqrt{a_1b_1}-(|u_2|+|u_3|)p.
\end{aligned}
$$
Summing up the above two inequalities, we have
$$
\begin{aligned}
&\sum_{i=1}^4(\sqrt{s_it_i}\sqrt{a_ib_i}-|u_i||z_i|)\\
&\phantom{XX}\ge (\sqrt{s_1t_1}+\sqrt{s_4t_4}-|u_1|-|u_4|)\sqrt{a_1b_1}
+(\sqrt{s_2t_2}+\sqrt{s_3t_3}-|u_2|-|u_3|)p,
\end{aligned}
$$
which is nonnegative by the inequalities $W_1[1,4]$ and $W_1[2,3]$. Therefore, we have
$$
\begin{aligned}
\frac 12\lan \xx(s,t,u),\xx(a,b,z)\ran
&=\frac 12\sum_{i=1}^4(s_ia_i+t_ib_i+2\re(u_iz_i))\\
&\ge \sum_{i=1}^4(\sqrt{s_it_i}\sqrt{a_ib_i}-|u_i||z_i|)\ge 0,
\end{aligned}
$$
which completes the proof.
\end{proof}

\section{Examples}

In order to get analytic examples distinguishing the convex cones in the inequalities
(\ref{dist-ineq}) and (\ref{dist-ineq-2}), we consider GHZ diagonal states
$$
\varrho_{0,0}=\frac 18\xx\left(\begin{matrix}1&1&1&1\\0&0&0&0\end{matrix}\right),\quad
\varrho_{1,0}=\frac 18\xx\left(\begin{matrix}1&2&0&1\\1&0&0&1\end{matrix}\right),\quad
\varrho_{0,1}=\frac 1{12}\xx\left(\begin{matrix}2&1&1&2\\2&1&1&0\end{matrix}\right),
$$
and define
$$
\begin{aligned}
\varrho_{s,t}
&=(1-s-t)\varrho_{0,0}+s\varrho_{1,0}+t\varrho_{0,1}\\
&=\frac 18\xx\left(\begin{matrix}
1+\frac 13 t & 1+s-\frac 13t &1-s-\frac 13 t &1+\frac 13 t\\
s+\frac 43 t & \frac 23 t  &\frac 23 t& s\end{matrix}\right),
\end{aligned}
$$
for real numbers $s$ and $t$.
We consider the convex set $\mathbb D$ of all three qubit states, which is a $63$ affine dimensional convex body.
We slice $\mathbb D$ by the $2$-dimensional plane $\Pi$ determined by
$\varrho_{0,0}$, $\varrho_{1,0}$ and $\varrho_{0,1}$ to get the pictures for various convex sets.
We see that $\varrho_{s,t}$ is a state if and only if
$$
\begin{aligned}
\textstyle{
|s+\frac 43 t|\le 1+\frac13t,\quad
|\frac 23 t|\le 1+s-\frac 13 t,\quad
|\frac 23 t|\le 1-s-\frac 13t,\quad
|s|\le 1+\frac 13 t
}
\end{aligned}
$$
if and only if $(s,t)$ belongs to the region
$$
\textstyle{
R=\{(s,t): s+t\le 1,\ -s+t\le 1,\ -s-\frac 53t\le 1,\ s-\frac 13 t\le 1\}
}
$$
which is a quadrilateral on the $st$-plane with the four vertices
$(1,0)$, $(\textstyle\frac 23, -1)$, $(-1,0)$ and $(0,1)$.
Therefore, the $2$-dimensional convex body $\mathbb D\cap\Pi$ is also a quadrilateral with the vertices
\begin{equation}\label{vertex}
\varrho_{0,1},\quad
\varrho_{1,0},\quad
\varrho_{\frac 23,-1}=\frac 1{12}\xx\left(\begin{matrix}1&3&1&1\\-1&-1&-1&1\end{matrix}\right),\quad
\varrho_{-1,0}=\frac 18\xx\left(\begin{matrix}1&0&2&1\\-1&0&0&-1\end{matrix}\right).
\end{equation}
See Figure 1.

\begin{figure}
\begin{center}
\setlength{\unitlength}{5 truecm}
\begin{picture}(2,2)
\thinlines
\drawline(-0.1,1)(2.1,1)
\drawline(1,0)(1,2.05)
\drawline(2,1)(1,2)(0,1)(1.667,0)(2,1)
\dottedline{0.01}(0.818181818,1.818181818)(0.6,1.6)(0.454545455,0.727272727)(1.666666667,0)
         (1.857142857,0.571428571)(0.818181818,1.818181818)
\dottedline{0.01}(1.4,1.6)(0.6,1.6)(0.333333333,1.333333333)(0.714285714,0.571428571)(1.666666667,0)(1.4,1.6)
\thicklines
\drawline(1.4,1.6)(0.818181818,1.818181818)(0.230769231,1.230769231)(0.454545455,0.727272727)
         (1.666666667,0)(1.857142857,0.571428571)(1.4,1.6)
\drawline(1.5,1)(1,1.6)(0.6,1.6)(0.5,1)(0.714285714,0.571428571)(1.666666667,0)(1.5,1)
\put(2,1){\circle{0.03}}\put(2.01,1.03){$(1,0)$}
\put(1,2){\circle{0.03}}\put(1.015,2.01){$(0,1)$}
\put(0,1){\circle{0.03}}\put(-0.225,1.03){$(-1,0)$}
\put(1.4,1.6){\circle*{0.03}}\put(1.405,1.62){$(\frac 25,\frac 35)$}
\put(0.335,1.62){$(-\frac 25,\frac 35)$}
\put(0.085,1.37){$(-\frac 23,\frac 13)$}
\put(1.46,0.92){$\frac 12$}
\put(0.52,0.92){$-\frac 12$}
\put(0.86,0.34){$-\frac 35$}
\put(0.434545455,0.4872727){$(-\frac {2}{7},-\frac 3{7})$}
\put(0.818181818,1.818181818){\circle*{0.03}}\put(0.498181818,1.848181818){$(-\frac 2{11},\frac 9{11})$}
\put(0.230769231,1.230769231){\circle*{0.03}}\put(-0.100769231,1.250769231){$(-\frac {10}{13},\frac 3{13})$}
\put(0.454545455,0.727272727){\circle*{0.03}}\put(0.074545455,0.6572727){$(-\frac {6}{11},-\frac 3{11})$}
\put(1.666666667,0){\circle{0.03}}\put(1.686666667,-0.04){$(\frac {2}{3},-1)$}
\put(1.857142857,0.571428571){\circle*{0.03}}\put(1.877142857,0.521428571){$(\frac {6}{7},-\frac 3{7})$}
\put(0.45,1.3){$\ccc$}
\put(1.3,1.4){$\ccc$}
\put(0.75,1.65){$\bbb$}
\put(0.5,0.78){$\bbb$}
\put(1.67,0.5){$\bbb$}
\put(1.67,1.04){$\frac 23$}
\end{picture}
\end{center}
\caption{The difference between the whole quadrilateral $R$ and the bigger hexagon $H_1$ shows us that the distributive rules do not hold.
The states $\varrho_{s,t}$ in the region labeled by $\bbb$ and $\ccc$ belong to $\bbb\setminus\ccc$ and $\ccc\setminus\bbb$, respectively.
The smaller hexagon represents the convex set $\aaa\meet\bbb\meet\ccc$ which consists of PPT states.}
\end{figure}
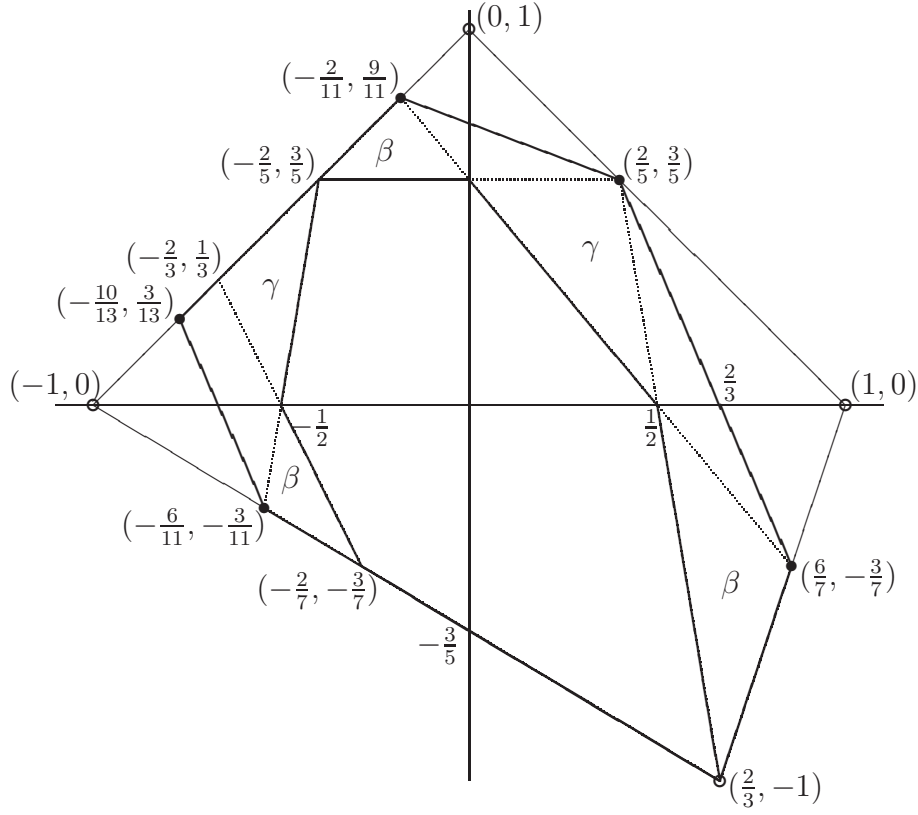

It is easily checked by $S_1[1,4]$ and $S_1[2,3]$
that the four states in (\ref{vertex}) belong to the convex set $\aaa$,
and so the convex set $\aaa$ on the plane $\Pi$ is represented by the quadrilateral $R$ itself.
Using $S_1[1,3]$, $S_1[2,4]$ and $S_1[1,2]$, $S_1[3,4]$, it is also easy to characterize $(s,t)$
such that $\varrho_{s,t}$ belongs to to $\bbb$ and $\ccc$, respectively.
One may check that
$\varrho_{s,t}\in\bbb$ if and only if it is a state and satisfies both inequalities $2s+\frac 53 t\le 1$
and $-2s+\frac 13 t\le 1$. Therefore, the region for $\bbb$ on the plane $\Pi$ is a pentagon
with vertices
$$
\textstyle{
(-\frac 25, \frac 35),\quad
(-\frac 2{11}, \frac 9{11}),\quad
(\frac 67, -\frac 37),\quad
(\frac 23, -1),\quad
(-\frac 6{11}, -\frac 3{11}).
}
$$
We also see that the region for $\ccc$ is determined by
$\frac 53t\le 1$, $-2s-t\le 1$ and $2s+\frac 13t\le 1$. This is also a pentagon with vertices
$$
\textstyle{
(-\frac 25, \frac 35),\quad
(\frac 25, \frac 35),\quad
(\frac 23, -1),\quad
(-\frac 27, -\frac 37),\quad
(-\frac 23, \frac 13).
}
$$

It is clear that the region for $\aaa\join\bbb$ or $\ccc\join\aaa$ on the plane $\Pi$ occupies all of the quadrilateral $R$.
One may also easily check by $S_2[1,4]$
that the four states in (\ref{vertex}) belong to $\bbb\join\ccc$,
and so the region for $\bbb\join\ccc$ coincides with the quadrilateral $R$. More precisely, the convex sets
$$
(\aaa\join\bbb)\cap\Pi=(\bbb\join \ccc)\cap\Pi=(\ccc\join\aaa)\cap\Pi=\aaa\cap\Pi
$$
are represented by the quadrilateral $R$.
The whole quadrilateral $R$ in Figure 1 thus represents the regions for the following convex sets
$$
\aaa,\quad \aaa\join\bbb,\quad \bbb\join\ccc,\quad
\ccc\join\aaa,\quad \aaa\join\bbb\join\ccc,\quad  \aaa\meet(\bbb\join\ccc),\quad \aaa\join(\bbb\meet\ccc)
$$
on the $st$-plane.
It should be noted that they are strictly bigger than the convex hull generated by $\bbb\cap\Pi$ and $\ccc\cap\Pi$.
For example, the state $\varrho_{1,0}\in\Pi$ in Figure 1 belongs to the convex hull $\bbb\join \ccc$,
but it is not a mixture of states in $\bbb\cap\Pi$ and $\ccc\cap\Pi$.
In fact, if $\varrho_{1,0}=\varrho_1+\varrho_2$ with $\varrho_1\in\bbb$ and $\varrho_1\in\ccc$ then one can easily see that
the {\sf X}-parts of $\varrho_1$ and $\varrho_2$ should be of the form
$$
\frac 18\xx\left(\begin{matrix}0&1&0&1\\0&*&0&1\end{matrix}\right)
\qquad {\text{\rm and}}\qquad
\frac 18\xx\left(\begin{matrix}1&1&0&0\\1&*&0&0\end{matrix}\right),
$$
respectively. Therefore, they never belong to the plane $\Pi$.

Now, we use Theorem \ref{join-statestep3} to find the region for the convex set $\bbb\join(\ccc\meet\aaa)$.
For pairs $\{i,j\}=\{1,2\},\{1,3\},\{1,4\},\{2,3\},\{2,4\},\{3,4\}$, the form
$8\sqrt{a_ib_i}+8\sqrt{a_jb_j}$ for the state $\varrho_{s,t}$ has the values
$$
2+s,\quad 2-s,\quad 2+\textstyle\frac 23 t,\quad
2-\textstyle\frac 23 t,\quad 2+s,\quad 2-s,
$$
respectively. On the other hands, $8|z_i|+8|z_j|$ becomes
$$
|s+\tfrac 43t|+\tfrac 23|t|,\quad
|s+\tfrac 43t|+\tfrac 23|t|,\quad
|s+\tfrac 43t|+|s|,\quad
\tfrac 43|t|,\quad
|s|+\tfrac 23|t|,\quad
|s|+\tfrac 23|t|.
$$
Therefore, we see that a state $\varrho_{s,t}$ belongs to $\bbb\join(\ccc\meet\aaa)$ if and only if
it belongs to $\ccc\join(\aaa\meet\bbb)$ if and only if the inequality
$$
\begin{aligned}
\min\{
2+s,\ 2-s,\ &2+\textstyle\frac 23 t,\ 2-\textstyle\frac 23 t\}\\
&\ge
\max\{
|s+\tfrac 43t|+\tfrac 23|t|,\
|s+\tfrac 43t|+|s|,\
\tfrac 43|t|,\
|s|+\tfrac 23|t|\}
\end{aligned}
$$
holds. One may check a point $(s,t)\in R$ satisfies this inequality if and only if
$$
\tfrac 12 s+\tfrac 43 t\le 1,\qquad
\tfrac 32s+\tfrac 23t\le 1,\qquad -\tfrac 32 s-\tfrac 23t\le 1,\qquad
$$
This region is represented by the bigger hexagon $H_1$ with the vertices
$$
\textstyle{
(-\frac 2{11},\frac 9{11}),\quad
(\frac 25,\frac 35),\quad
(\frac 67,-\frac 37),\quad
(\frac 23,-1),\quad
(-\frac 6{11},-\frac 3{11}),\quad
(-\frac {10}{13},\frac 3{13})
}
$$
in Figure 1.
Therefore, the difference $R\setminus H_1$ consisting of three triangles gives us examples in
\begin{equation}\label{gap-2}
(\bbb\join\ccc)\meet(\bbb\join\aaa)\ \setminus\ \bbb\join(\ccc\meet\aaa),
\end{equation}
which shows that the strict inequality holds in (\ref{dist-ineq-2}).

In order to consider the inequality (\ref{dist-ineq}),
we first note the inequality
\begin{equation}\label{modular}
(\aaa\meet\bbb)\join(\aaa\meet\ccc)\le \aaa\meet(\bbb\join(\ccc\meet\aaa))\le \aaa\meet(\bbb\join\ccc),
\end{equation}
which holds in general.
By the first inequality, the region for
$(\aaa\meet\bbb)\join(\aaa\meet\ccc)$ is also contained in $H_1$. In fact, it fills up all of $H_1$.
Indeed, it is clear that five vertices of $H_1$ belong to
$(\aaa\meet\bbb)\join(\aaa\meet\ccc)$ except for $(-\frac {10}{13},\frac 3{13})$ by $S_1[i,j]$. We also
see that
$$
52\,\varrho_{-\frac {10}{13},\frac 3{13}}
=\xx\left(\begin{matrix}7&1&11&7\\-3&1&1&-5\end{matrix}\right)
=\xx\left(\begin{matrix}3&1&3&3\\-3&1&1&-1\end{matrix}\right)
+\xx\left(\begin{matrix}4&0&8&4\\0&0&0&-4\end{matrix}\right)
$$
belongs to $(\aaa\meet\bbb)\join(\aaa\meet\ccc)$. Therefore, this coincides with
$\aaa\meet (\bbb\join(\ccc\meet\aaa))$ for states $\varrho_{s,t}$,
that is, the first inequality in (\ref{modular}) becomes an identity on the plane $\Pi$.
Therefore, the difference $R\setminus H_1$ again gives rise to examples in
\begin{equation}\label{gap}
\aaa\meet(\bbb\join\ccc)\ \setminus\ (\aaa\meet\bbb)\join(\aaa\meet\ccc),
\end{equation}
for the strict inequality in (\ref{dist-ineq}).

Now, we proceed to show that the gaps (\ref{gap-2}) and (\ref{gap}) arising in the distributive inequalities have nonzero volume.
We first note that the convex set $\mathbb S$ of all fully separable three qubit states generates the same affine manifold
as the convex set $\mathbb D$ of all three qubit states. See the discussion at the end of Section 7 in \cite{ha-han-kye}.
Therefore, all the convex sets in (\ref{gap-2}) and (\ref{gap}) generate the same affine manifold.
We also recall that a point $x_0$ of a convex set $C$ is called an {\sl interior point} of $C$
if it is an interior point of $C$ with respect to the affine space generated by $C$.
We note that the state $\varrho_{0,0}$ is a common interior point of the convex sets appearing in
(\ref{gap-2}) and (\ref{gap}) as well as $\mathbb S$ and $\mathbb D$.
If we consider the line segment
$x_t=(1-t)x_0+tx_1$ between an interior point $x_0$ of a convex set $C$
and an arbitrary point $x_1\in C$ then $x_t$ is also an interior point of $C$
for $0<t<1$. See \cite[Lemma 2.3]{kye-canad}.
Therefore, every interior point of $[\aaa\meet(\bbb\join\ccc)]\cap \Pi$ is actually
an interior point of $\aaa\meet(\bbb\join\ccc)$.
For example, $\varrho_{\frac 23,0}$ is an interior point of $\aaa\meet(\bbb\join\ccc)$
which is a boundary point of $(\aaa\meet\bbb) \join (\aaa\meet\ccc)$.
From this, we may conclude that the difference (\ref{gap}) has the nonempty interior by \cite[Proposition 7.4]{ha-han-kye}.
The exactly same argument also shows that the difference (\ref{gap-2}) also has the nonempty interior.

We also consider the smaller hexagon $H_2$ with vertices
$$
\textstyle{
(-\frac 25,\frac 35),\quad
(0,\frac 35),\quad
(\frac 12,0),\quad
(\frac 23,-1),\quad
(-\frac 27,-\frac 37),\quad
(-\frac 12,0)
}
$$
in Figure 1, which is the region
for $\varrho_{s,t}$ in $\bbb\meet\ccc$.
This represents also $\aaa\meet\bbb\meet\ccc$. Note that an $\xx$-shaped state belongs to $\aaa\meet\bbb\meet\ccc$ if and only if
it is of positive partial transpose by \cite[Theorem 5.3]{han_kye_optimal}.
Therefore, the region $H_2$ represents the region for PPT states for $\varrho_{s,t}$.

Recall that a lattice is called {\sl modular} if $x\le z$ implies $x\join(y\meet z)=(x\join y)\meet z$.
This is the case if and only if the modular identity
$$
(x\meet z)\join (y\meet z)= ((x\meet z)\join y)\meet z
$$
holds for every $x,y$ and $z$. Every distributive lattice is modular.  See  \cite{birkhoff,salii} for
elementary properties of modular lattices. We exhibit examples of states
showing that the first inequality in (\ref{modular}) is also strict,
to conclude that the lattice ${\mathcal L}$ is not modular.
To do this, we consider
$\varrho_1=\dfrac 1{12}\xx\left(\begin{matrix}2&1&1&2\\2&0&1&0\end{matrix}\right)$,
and put
$$
\varrho_t=(1-t)\varrho_{0,0}+t\varrho_1
=
\frac 1{24}\xx\left(\begin{matrix} 3+t &3-t & 3-t &3+t\\ 4t &0 &2t &0\end{matrix}\right).
$$
We also consider
$W=\xx\left(\begin{matrix}0&1&1&0\\-1&0&-1&0\end{matrix}\right)$,
which satisfies the inequality $W_2[i,j]$ for
$\{i,j\}=\{1,2\}, \{3,4\}, \{1,3\},\{2,4\}$ and $W_3$.
Therefore, we see that $W$ belongs to the
convex cone
$(\aaa^\circ\join\bbb^\circ)\meet(\aaa^\circ\join\ccc^\circ)$,
which is the dual of the convex cone
$(\aaa\meet\bbb)\join(\aaa\meet\ccc)$. Now, we see that
$\lan W,\varrho_t\ran=\frac 1{24}(6-8t)\ge 0$ if and only if
$t\le\frac 34$. Therefore, we conclude that $\varrho_t$ does not
belong to $(\aaa\meet\bbb)\join(\aaa\meet\ccc)$ for
$\frac 34 < t\le 1$.  On the other hand, one can easily check that
$\varrho_1\in \aaa\meet(\bbb\join(\ccc\meet\aaa))$ by
Theorem \ref{join-statestep3}.  Therefore, we see that $\varrho_t$
also belongs to the same cone. In fact, $\varrho_t$ is an interior
point of the cone $\aaa\meet(\bbb\join(\ccc\meet\aaa))$
for $0\le t<1$, because $\varrho_{0,0}$ is an interior point. Hence,
we also see that the gap for the first inequality in (\ref{modular})
has also nonzero volume.

Finally, we also characterize the full separability for the states we are considering.
To do this, we summarize the results in \cite{{guhne_pla_2011},{han_kye_GHZ},{kay_2011}}.
See also \cite{{chen_han_kye},{ha-han-kye}}. For a given GHZ diagonal state $\varrho=\xx(a,a,c)$
with $a,c\in\mathbb R^4$, we consider the following:
$$
\begin{aligned}
\lambda_5 = 2(+c_1+c_2+c_3+c_4),\quad &\lambda_6 = 2(-c_1-c_2+c_3+c_4),\\
\lambda_7 = 2(-c_1+c_2-c_3+c_4),\quad &\lambda_8 = 2(-c_1+c_2+c_3-c_4),\\
t_1=c_1(-c_1^2+c_2^2+c_3^2+c_4^2)-2c_2c_3c_4,\quad
&t_2=c_2(+c_1^2-c_2^2+c_3^2+c_4^2)-2c_1c_3c_4,\\
t_3=c_3(+c_1^2+c_2^2-c_3^2+c_4^2)-2c_1c_2c_4,\quad
&t_4=c_4(+c_1^2+c_2^2+c_3^2-c_4^2)-2c_1c_2c_3.
\end{aligned}
$$
When all the following inequalities
\begin{equation}\label{kufl}
\lambda_5\lambda_6\lambda_7\lambda_8 > 0,\qquad
t_1t_4\lambda_6\lambda_7 < 0,\qquad
t_2t_3\lambda_5\lambda_8 > 0
\end{equation}
hold, the state $\varrho=\xx(a,a,c)$ is fully separable if and only if
the inequality
\begin{equation}\label{sufficient}
\min \{a_1,a_2,a_3,a_4\} \ge {\sqrt{(\lambda_5\lambda_6+\lambda_7\lambda_8)
(\lambda_5\lambda_7+\lambda_6\lambda_8)
(\lambda_5\lambda_8+\lambda_6\lambda_7)}
\over
{8\sqrt{\lambda_5\lambda_6\lambda_7\lambda_8}}}
\end{equation}
is satisfied. In the other cases, $\varrho$ is fully separable if
and only if it is of PPT. For the state $\varrho_{s,t}$, the
conditions (\ref{kufl}) are given by
$$
s(3s+4t)<0,\qquad (9s^2+18st+4t^2)(9s^2+6st-4t^2)> 0.
$$
We note that a point $(s,t)$ on the line $t=as$ satisfies this condition if and only if $-\frac 34(3+\sqrt 5)\le a<-\frac 34$.
On the other hands, the square of right side
of (\ref{sufficient}) is given by
$$
\frac{t^2(9s^2+12st-4t^2)}{432s(3s+4t)}
$$
respectively.
See Figure 2.

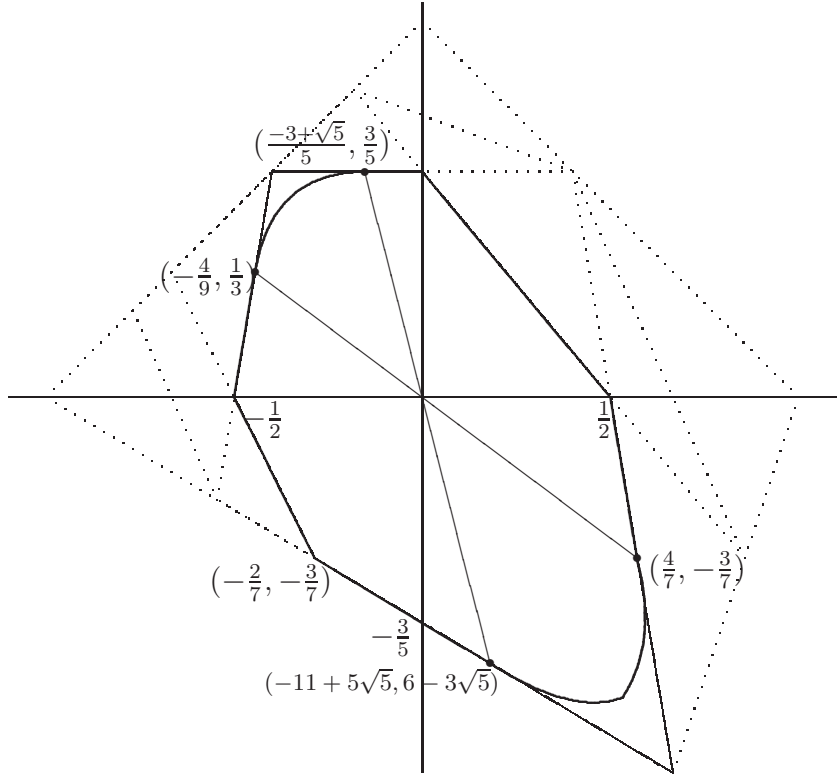
\begin{figure}
\begin{center}
\setlength{\unitlength}{5 truecm}
\begin{picture}(2,2)
\thinlines
\drawline(-0.1,1)(2.1,1)
\drawline(1,0)(1,2.05)
\dottedline{0.03}(2,1)(1,2)(0,1)(1.667,0)(2,1)
\dottedline{0.03}(0.818181818,1.818181818)(0.6,1.6)(0.454545455,0.727272727)(1.666666667,0)
         (1.857142857,0.571428571)(0.818181818,1.818181818)
\dottedline{0.03}(1.4,1.6)(0.6,1.6)(0.333333333,1.333333333)(0.714285714,0.571428571)(1.666666667,0)(1.4,1.6)
\dottedline{0.03}(1.4,1.6)(0.818181818,1.818181818)(0.230769231,1.230769231)(0.454545455,0.727272727)
         (1.666666667,0)(1.857142857,0.571428571)(1.4,1.6)
\dottedline{0.005}(1.5,1)(1,1.6)(0.6,1.6)(0.5,1)(0.714285714,0.571428571)(1.666666667,0)(1.5,1)
\put(1.46,0.92){$\frac 12$}
\put(0.52,0.92){$-\frac 12$}
\put(0.86,0.34){$-\frac 35$}
\put(0.434545455,0.4872727){$(-\frac {2}{7},-\frac 3{7})$}
\drawline(0.555555556,1.333333333)(1.571428571,0.571428571)
\drawline(0.847213595,1.6)(1.180339887,0.291796068)
\put(0.555555556,1.333333333){\circle*{0.02}}
\put(1.571428571,0.571428571){\circle*{0.02}}
\put(0.847213595,1.6){\circle*{0.02}}
\put(1.180339887,0.291796068){\circle*{0.02}}
\dottedline{0.005}(1,1.6)(0.6,1.6)(0.5,1)(0.714285714,0.571428571)(1.666666667,0)(1.5,1)
\thicklines
\drawline(0.847213595, 1.6)(1,1.6)(1.5,1)(1.571428571, 0.571428571)
\drawline(0.555555556, 1.333333333)(0.5,1)(0.714285714,0.571428571)(1.180339887, 0.291796068)
\qbezier(0.555555556, 1.333333333)(0.6,1.6)(0.847213595, 1.6)
\qbezier(1.571428571, 0.571428571)(1.63,0.35) (1.533333333, 0.2)
\qbezier(1.180339887, 0.291796068)(1.4,0.15) (1.533333333, 0.2)
\put(0.3,1.3){$(-\frac 49,\frac 13)$}
\put(0.55,1.65){$(\frac{-3+\sqrt 5}{5},\frac 35)$}
\put(1.6,0.52){$(\frac 47,-\frac 37)$}
\put(0.58,0.22){\scriptsize{$(-11+5\sqrt 5, 6-3\sqrt 5)$}}
\end{picture}
\end{center}
\caption{Two line segments through the origin are given by the conditions (\ref{kufl}), and
two curves surrounding the region of full separability are given by (\ref{sufficient}).}
\end{figure}

\section{General multi-partite cases}

Now, we turn our attention to general multi-partite system $M_{A_1}\ot M_{A_2}\otimes \cdots\otimes M_{A_n}$.
For a given partition $\Pi$ of local systems $\{A_1,A_2,\dots,A_n\}$, we denote by $\aaa_\Pi$
the convex cone consisting of all partially separable (unnormalized) states with respect to the partition $\Pi$.
When $M_{A_i}$ is the $d_i\times d_i$ matrices, we denote by ${\mathcal L}_{d_1,d_2,\dots,d_n}$
the lattice generated by $\alpha_\Pi$ through all nontrivial partitions $\Pi$
with respect to two operations, intersection and convex hull. Therefore, elements of this lattice
are convex cones sitting in the real vector space of all $d_1d_2\cdots d_n\times d_1d_2\cdots d_n$ Hermitian matrices.
So far, we have considered the lattice ${\mathcal L}_{2,2,2}$. We close this section by showing that the lattice
${\mathcal L}_{d_1,d_2,\dots,d_n}$ also violates distributivity and modularity.

For $d \ge 2$, we consider the canonical embedding
$$
\iota_d : \varrho \in M_2 \mapsto \varrho \oplus {\rm tr}(\varrho) 1_{d-2} \in M_d
$$
and the compression
$$
Q_d : \varrho \in M_d \mapsto \begin{pmatrix} I_2& 0 \end{pmatrix} \varrho \begin{pmatrix} I_2 \\ 0 \end{pmatrix} \in M_2.
$$
Here, the trace is a normalized one.
Then, the composition $Q_d \circ \iota_d$ is the identity map.
In other words, $M_2$ is unital completely positive (u.c.p.) complemented in $M_n$.
This means that every local property of $M_p \otimes M_q \otimes M_r$ is hereditary to $M_2 \otimes M_2 \otimes M_2$.
We write
$$
\begin{aligned}
\iota&=\iota_p \otimes \iota_q \otimes \iota_r : M_2 \otimes M_2 \otimes M_2 \to M_p \otimes M_q \otimes M_r\\
Q&=Q_p \otimes Q_q \otimes Q_r : M_p \otimes M_q \otimes M_r \to M_2 \otimes M_2 \otimes M_2
\end{aligned}
$$
for brevity.

We retain the notations $\aaa$, $\bbb$ and $\ccc$ for the three qubit case, that is,
$\alpha=\alpha_{A{\text{\rm -}}BC}$,
$\beta=\alpha_{B{\text{\rm -}}CA}$ and $\gamma=\alpha_{C{\text{\rm -}}AB}$ generate the lattice ${\mathcal L}_{2,2,2}$.
To make clear, we will use notations $\tilde\aaa$, $\tilde\bbb$ and $\tilde\ccc$ for generators of the lattice
${\mathcal L}_{p,q,r}$. It is easily seen that $\varrho\in\aaa$ implies $\iota(\varrho)\in\tilde\aaa$, and
$\omega\in\tilde\aaa$ implies $Q(\omega)\in\aaa$, and similarly for $\bbb$ and $\ccc$.
In order to show that ${\mathcal L}_{p,q,r}$ violates the distributive rules, we first take
$$
\varrho \in \left(\alpha \meet (\beta \join \gamma)\right) ~\backslash~ (\alpha \meet \beta) \join (\alpha \meet \gamma)
$$
in  $M_2 \otimes M_2 \otimes M_2$.
We can write $\varrho=\varrho_2+\varrho_3$ for $\varrho_2 \in \beta$ and $\varrho_3 \in \gamma$.
We have
$$
\iota(\varrho) \in \tilde\alpha, \qquad \iota(\varrho_2) \in \tilde\beta, \qquad \iota(\varrho_3) \in \tilde\gamma
$$
in $M_p \otimes M_q \otimes M_r$.
Thus, $\iota(\varrho)$ belongs to $\tilde\alpha \meet (\tilde\beta \join \tilde\gamma)$ in $M_p \otimes M_q \otimes M_r$.
Assume to the contrary that the distributive rule holds in $M_p \otimes M_q \otimes M_r$.
Then, $\iota(\varrho)$ belongs to $(\tilde\alpha \meet \tilde\beta) \join (\tilde\alpha \meet \tilde\gamma)$, and so
we can write $\iota(\varrho)=\omega_2+\omega_3$ with $\omega_2 \in \tilde\alpha \meet \tilde\beta$ and $\omega_3 \in \tilde\alpha \meet \tilde\gamma$.
Therefore, we have
$$
\varrho = Q \circ \iota (\varrho) = Q(\omega_2) + Q(\omega_3) \in (\alpha \meet \beta) \join (\alpha \meet \gamma),
$$
which is a contradiction. The same argument can be applied to the other distributive rule.
It should be noted that $\omega_2$ and $\omega_3$ themselves need not belong to the image of $\iota$ in the above argument.
Non-modularity also follows from that of three qubit system in the same fashion.

For general multi-partite cases, we consider the canonical embedding
$$
\iota : x \in M_{d_1} \otimes M_{d_2} \otimes M_{d_3} \mapsto x \otimes 1
   \in M_{d_1} \otimes M_{d_2} \otimes M_{d_3} \otimes M_{d_4} \otimes \cdots \otimes M_{d_n}
$$
and the (normalized) partial trace
$$
{\rm id} \otimes {\rm tr} : M_{d_1} \otimes M_{d_2} \otimes M_{d_3} \otimes M_{d_4} \otimes \cdots \otimes M_{d_n}
     \to M_{d_1} \otimes M_{d_2} \otimes M_{d_3}.
$$
Then, the composition $({\rm id} \otimes {\rm tr}) \circ \iota$ is the identity map.
In other words, $M_{d_1} \otimes M_{d_2} \otimes M_{d_3}$ is u.c.p. complemented in
$M_{d_1} \otimes M_{d_2} \otimes M_{d_3} \otimes M_{d_4} \otimes \cdots \otimes M_{d_n}$.
This means that every local property of the latter
is hereditary to $M_p \otimes M_q \otimes M_r$.
By the exactly same argument as before, one may check that the lattice ${\mathcal L}_{d_1,\dots, d_1}$ also violates
distributivity and modularity.

\section{Summary and further questions}

In this paper, we have considered the lattice ${\mathcal L}$
generated by three basic convex sets $\aaa$, $\bbb$ and $\ccc$
consisting of all $A$-$BC$ biseparable, $B$-$CA$ biseparable and
$C$-$AB$ biseparable three qubit states, respectively, with respect to the
operations of convex hull $\join$ and intersection $\meet$. In this
way, we may consider convex sets of partially separable states
obtaining by arbitrary convex hulls and intersections of $\aaa$, $\bbb$ and $\ccc$, and the whole
structure of partial separability may be revealed by mathematical
properties of the lattice $\mathcal L$. For general theory for lattices, we refer to
the monographs \cite{birkhoff}, \cite{freese} and \cite{salii}.

We gave the negative answer to the first natural question asking if
this lattice is distributive.
The lattice $\mathcal L$ is not even modular.
Another interesting question is to ask if the lattice ${\mathcal L}$ has
infinitely many elements. We conjecture this is the case. This
means that there are infinitely many kinds of partial separability
and partial entanglement. It is known that a free lattice with three generators must have infinitely many
elements. In this regard, it would be interesting to know if the lattice $\mathcal L$ is free or not.

The next question is whether the lattice ${\mathcal L}$ is complemented or not.
A lattice is called {\sl complemented} if every element $x$ has a complement $y$
which satisfies $x\meet y=0$ and $x\join y=1$, where $0$ and $1$ denote the least and greatest elements, respectively.
The least and the greatest elements of ${\mathcal L}$ are given by $\aaa\meet\bbb\meet\ccc$ and
$\aaa\join\bbb\join\ccc$, respectively. They represent the set of
all fully biseparable and biseparable states, respectively. Especially, we would like to ask
if $\aaa$ has a complement, that is, we ask
if there exist $\sigma\in{\mathcal L}$ such that
$\aaa\meet\sigma=\aaa\meet\bbb\meet\ccc$ and
$\aaa\join\sigma=\aaa\join\bbb\join\ccc$. Recall that the set of all
closed subspaces of a Hilbert space makes a lattice, the {\sl
subspace lattice}, with respect to the closed linear hull and
intersection. This plays an important role in quantum logic and
theory of operator algebras. We note that the subspace
lattice is non-distributive, but complemented.

Finally, we ask how the lattice ${\mathcal L}_{d_1,\dots,d_n}$ depends on the dimensions of the local systems.
In the case of tri-partite systems, we are asking if two lattices ${\mathcal L}_{2,2,2}$ and ${\mathcal L}_{p,q,r}$
are isomorphic to each other.

\smallskip

Acknowledgement: Both KHH and SHK were partially supported by the grant NRF-2017R1A2B4006655, Korea.
SzSz was supported by the National Research, Development and
Innovation Fund of Hungary within the {\it Researcher-initiated
Research Program} (project Nr:~NKFIH-K120569) and within the {\it
Quantum Technology National Excellence Program} (project
Nr:~2017-1.2.1-NKP-2017-00001), by the Ministry for Innovation and
Technology within the {\it {\'U}NKP-19-4 New National Excellence
Program}, and by the Hungarian Academy of Sciences within the {\it
J{\'a}nos Bolyai Research Scholarship} and the {\it ``Lend{\"u}let''
Program}.


\end{document}